%% file: main.tex
\documentclass[conference]{IEEEtran}
\IEEEoverridecommandlockouts
\usepackage{amsmath,graphicx}
\input{packages}

\input{macros}

\title{Cross-Channel Unlabeled Sensing over a \\ Union of Signal Subspaces}

\author{
\IEEEauthorblockN{
  Taulant Koka\IEEEauthorrefmark{1},
  Manolis C. Tsakiris\IEEEauthorrefmark{2},
  Benjamín Béjar Haro\IEEEauthorrefmark{3},
  Michael Muma\IEEEauthorrefmark{1}
}
\IEEEauthorblockA{\IEEEauthorrefmark{1}%
  Robust Data Science Group, Technische Universität Darmstadt, 
  Darmstadt, Germany\\
}
\IEEEauthorblockA{\IEEEauthorrefmark{2}%
  Academy of Mathematics and Systems Science, Chinese Academy of Sciences, 
  Beijing, China\\
}
\IEEEauthorblockA{\IEEEauthorrefmark{3}%
  Swiss Data Science Center, Paul Scherrer Institute, 
  Villigen, Switzerland\\
}
\thanks{The work of T. Koka and M. Muma has been funded by the ERC Starting Grant ScReeningData, Grant No. 101042407. { M. C. Tsakiris acknowledges the support of the National Key R\&D Program of China (2023YFA1009402) and of the CAS Project for Young Scientists in Basic Research (YSBR-034).}}

\thanks{{\color{red}© 2025 IEEE. Personal use of this material is permitted. Permission from IEEE must be obtained for all other uses, in any current or future media, including reprinting/republishing this material for advertising or promotional purposes, creating new collective works, for resale or redistribution to servers or lists, or reuse of any copyrighted component of this work in other works. Accepted to ICASSP 2025. DOI: \texttt{10.1109/ICASSP49660.2025.10888212}\par\vskip 1em}}}

\begin{document}

\maketitle

\begin{abstract}
Cross-channel unlabeled sensing addresses the problem of recovering a multi-channel signal from measurements that were shuffled across channels. This work expands the cross-channel unlabeled sensing framework to signals that lie in a union of subspaces. The extension allows for handling more complex signal structures and broadens the framework to tasks like compressed sensing. These mismatches between samples and channels often arise in applications such as whole-brain calcium imaging of freely moving organisms or multi-target tracking. We improve over previous models by deriving tighter bounds on the required number of samples for unique reconstruction, while supporting more general signal types. The approach is validated through an application in whole-brain calcium imaging, where organism movements disrupt sample-to-neuron mappings. This demonstrates the utility of our framework in real-world settings with imprecise sample-channel associations, achieving accurate signal reconstruction. 
\end{abstract}

\section{Introduction}
\label{sec:intro}

\input{Sections/introduction}
\section{Cross-Channel Unlabeled Sensing: extension to unions of subspaces}
 \label{sec:meth}

\input{Sections/methodology-2}

\section{Numerical Simulations}
\label{sec:simul}
\input{Sections/simulations}

\section{Conclusion}
\label{sec:conc}
In this work, we extended the framework of cross-channel unlabeled sensing to accommodate signals residing in distinct subspaces, thereby allowing for more complex signal structures. Our theoretical contributions include establishing conditions for unique signal recovery and showing that sparse signal models are specific instances of the framework. We validated the practical utility of the cross-channel unlabeled sensing framework for sparse signal structures through experiments on whole-brain calcium imaging data from \emph{Drosophila}, showcasing its effectiveness in reconstructing shuffled signals under realistic conditions. The results highlight the framework's potential to handle real-world scenarios with imprecise sample-channel associations, making it a valuable tool for applications in neuroscience and beyond.
\bibliographystyle{IEEEbib}
\bibliography{bibliography.bib}

\end{document}

%% file: packages.tex
\usepackage{amsmath,amsfonts,amssymb,amsthm,amscd,array,latexsym}
\usepackage[caption=false,font=normalsize,labelfont=sf,textfont=sf]{subfig}
\usepackage{textcomp}
\usepackage{stfloats}
\usepackage{url}
\usepackage{verbatim}
\usepackage{nccmath}
\usepackage{graphicx}
\usepackage{cite}
\usepackage[hidelinks]{hyperref}      
\usepackage[resetlabels]{multibib}
\usepackage{url}            
\usepackage{booktabs}       
\usepackage{nicefrac}      
\usepackage{placeins}
\usepackage{multirow}
\usepackage{bm}
\usepackage{siunitx}
\usepackage{enumitem}
\usepackage[ruled]{algorithm}
\usepackage[]{algorithmicx}
\usepackage[noend]{algpseudocode}
\usepackage[dvipsnames]{xcolor}
\usepackage[percent]{overpic}
\usepackage{standalone}
\usepackage{xcolor}
\usepackage{soul}
\usepackage{enumerate}
\usepackage{mathrsfs}
\usepackage{overpic}
\usepackage{mathtools}
\usepackage{graphicx,subfig,wrapfig}
\usepackage{times}
\usepackage{psfrag,epsfig}
\usepackage{comment}
\usepackage{tabularx}
\usepackage{setspace}
\usepackage{color}
\usepackage[pagewise]{lineno}

%% file: macros.tex
\DeclareMathAlphabet{\mathpzc}{OT1}{pzc}{m}{it}
\definecolor{blue}{rgb}{0.0, 0.53, 0.74}
\renewcommand\thesection{\arabic{section}}
\renewcommand\thesubsection{\thesection.\arabic{subsection}}
\renewcommand\thesubsubsection{\thesubsection.\arabic{subsubsection}}

\newtheorem{theorem}{Theorem}

\newtheorem{definition}{Definition}
\newtheorem*{prop*}{Restricted Full Rank Property (RFRP)}
\newtheorem{corollary}{Corollary}
\newtheorem{remark}{Remark}

\newtheorem*{prf*}{Proof}

\def\Ce{\mathbb{C}}
\def\Ic{\mathcal{I}}

\def\Es{\mathscr{E}}

\DeclareMathOperator{\krank}{k-rank}
\DeclareMathOperator{\im}{im}

\DeclareFontFamily{U}{BOONDOX-calo}{\skewchar\font=45 }
\DeclareFontShape{U}{BOONDOX-calo}{m}{n}{
  <-> s*[1.05] BOONDOX-r-calo}{}
\DeclareFontShape{U}{BOONDOX-calo}{b}{n}{
  <-> s*[1.05] BOONDOX-b-calo}{}
\DeclareMathAlphabet{\mathcalboondox}{U}{BOONDOX-calo}{m}{n}
\SetMathAlphabet{\mathcalboondox}{bold}{U}{BOONDOX-calo}{b}{n}
\DeclareMathAlphabet{\mathbcalboondox}{U}{BOONDOX-calo}{b}{n}

%% file: Sections/introduction.tex
Reconstructing signals without precise knowledge of sample locations, often referred to as \emph{unlabeled sensing} or \emph{shuffled linear regression}, has recently gained significant attention in signal processing research \cite{unnikrishnan_unlabeled_2018,dokmanic_permutations_2019,tsakiris_algebraic-geometric_2020,slawski_pseudo-likelihood_2021,tsakiris_homomorphic_2019,ashwin_pananjady_linear_2018,hsu_linear_2017,abid_stochastic_2018,peng_linear_2020,elhami_unlabeled_2017,peng_homomorphic_2021,xu_uncoupled_2019,abbasi_r-local_2021,yao_unlabeled_2021,wang_estimation_2020,slawski_linear_2019,abid_linear_2017,TSAKIRIS2023210,onaran2022shuffled,abbasi_r-local_2022,Li2022}. This task typically involves recovering the correct order of samples and associated coefficients when the sensing matrix is known.
A related problem explored in \cite{Koka2024} introduced \emph{cross-channel unlabeled sensing}, where shuffling occurs across multiple channels of a multichannel signal rather than within individual channels. In this framework, signals reside within the same subspace, and samples may be swapped between channels, adding structure to the reconstruction task. The paper presented two key results: it formalized the cross-channel unlabeled sensing problem with signals confined to a single subspace, providing conditions for unique reconstruction. It then extended the theory to a specific sparse recovery scenario in continuous time for signals modeled as low-pass filtered streams of weighted Diracs.
Building on these foundations, our work generalizes the framework by allowing each channel signal to reside in its own distinct subspace, thus accommodating more complex and diverse signal structures. This extension significantly broadens the applicability of cross-channel unlabeled sensing beyond shared subspaces, making it suitable for a wider range of problems, including sparse recovery tasks like compressed sensing on and off the grid \cite{Donoho2006,Golbabaee2012,candes-fernandez-2012,tang-etal-2013}. A key improvement of this work is the derivation of a tighter bound on the number of required samples for unique reconstruction compared to Theorem 2 in \cite{Koka2024}, while at the same time allowing for more general signal models.
We validate our approach through experiments in whole-brain calcium imaging, a neuroscience application where matching measurements to specific neurons can be challenging due to the movement of the organism. The application highlights the practical utility of our framework in managing real-world scenarios with imprecise associations between samples and channels. 

\textbf{Main Contributions:} (i) We generalize the cross-channel unlabeled sensing problem to multi-channel signals residing in a union of subspaces, as formalized in Theorem~\ref{thm:CCULS}. (ii) We demonstrate that compressed sensing is a special instance of our generalized framework, as shown in Corollary~\ref{corr:1} and Corollary~\ref{corr:2}. (iii) We showcase the applicability of our approach in a practical neuroscience context involving whole-brain calcium imaging, where neural signals are recovered from shuffled traces.

\textbf{Notation:} We represent sets with calligraphic letters, tensors with bold calligraphic letters, and vectors and matrices with bold letters. An $N$-dimensional vector is expressed as $\bm x = ( x_0,\: x_1,\: \ldots,\: x_{N-1})^\top$, and an $N \times N$ matrix is represented by $\bm X = (\bm x_0 \quad \bm x_1 \: \cdots \: \bm x_{N-1})$. 
For indexed elements ${\Omega_1, \ldots, \Omega_M}$, the notation $\{\Omega_m\}_{m=1}^M$ is used for brevity. For subspaces $\mathcal{V}$ and $\mathcal{W}$, their dimensions are written as $\dim \mathcal{V}$ and $\dim \mathcal{W}$, their intersection as $\mathcal{V} \cap \mathcal{W}$, and their union as $\mathcal{V} \cup \mathcal{W}$. We define as $\mathcal{V}+\mathcal{W}=\{\bm v + \bm w|\bm v\in\mathcal{V},\:\bm w \in \mathcal{W}\}$ the sum of subspaces. The kernel of a linear map $\varrho: \mathcal{V} \rightarrow \mathcal{W}$ is denoted by $\ker(\varrho)$, and its image is denoted by $\im(\varrho)$. The Kronecker product is denoted by the operator $\otimes$.

%% file: Sections/methodology-2.tex
Let us consider a multi-channel signal $\bm x_m \in \Es_m$ for $m = 1, \dots, M$, where $M$ denotes the number of channels and each $\Es_m \subset \mathbb{C}^N$ is a $K_m$-dimensional linear subspace. Specifically, $\bm x_m$ can be expressed as $\bm x_m = \bm E_m \bm \beta_m$, where the columns of $\bm E_m \in \mathbb{C}^{N \times K_m}$ form a basis for $\Es_m$. Our focus is on signals $\bm y_m \in \mathbb{C}^N$, such that each row of the $N \times M$ matrix $\bm Y = (\bm y_1 \cdots \bm y_M)$ corresponds to the same entries as the rows of the $N \times M$ matrix $\bm X = (\bm x_1 \cdots \bm x_M)$, except for permutations across the channels, as illustrated in Fig.~\ref{fig:shuffled}. For $\bm E_1 =\dots=\bm E_M$, this setting coincides with the \emph{cross-channel unlabeled sensing} framework introduced in \cite{Koka2024}.  In that work, Definitions~\ref{def:shuffled_multi_channel}~and~\ref{dfn:generic} were established to formally define the concepts of shuffled multi-channel signals and the \hyperref[dfn:generic]{restricted full rank property (RFRP)}, which are fundamental to the theoretical foundation.

\begin{definition}[\textbf{Shuffled Multi-Channel Signal} \cite{Koka2024}]
\label{def:shuffled_multi_channel}
For every pair of indices $m,n\in\{1,\dots,M\}$, consider binary vectors $\boldsymbol{q}_{mn}\in{\{0,1\}}^N$ satisfying 

\begin{equation}
\label{eq:sum_q}
\sum_{m=1}^M\boldsymbol{q}_{mn}~=~\sum_{n=1}^M\boldsymbol{q}_{mn}=
    {\begin{pmatrix}
        1&1&\cdots&1
    \end{pmatrix}^\top}.
\end{equation}

We define $\bm y_m\in\mathbb{C}^N, \, m=1,\dots,M$, to be a shuffled multi-channel signal if it is related to the underlying multi-channel signal $\bm x_m~\in~\mathbb{C}^N, m=1,\dots,M$ by the following relation:

\begin{align}
\label{eq:permutation_model_msignals}
    \underbrace{\begin{pmatrix}
        \boldsymbol{y}_{1} \\
        \vdots \\
        \boldsymbol{y}_{M} \\
    \end{pmatrix}}_{\eqqcolon\bm y}
    =
    \underbrace{\begin{pmatrix}
        \boldsymbol{Q}_{11} &\cdots&\boldsymbol{Q}_{1M} \\
        \vdots & \ddots&\vdots \\
        \boldsymbol{Q}_{M1} &\cdots&\boldsymbol{Q}_{MM} \\
    \end{pmatrix}}_{\eqqcolon \bm \Pi }
    \underbrace{\begin{pmatrix}
        \boldsymbol{x}_{1} \\
        \vdots \\
        \boldsymbol{x}_{M} \\
    \end{pmatrix}}_{\eqqcolon\bm x}, 
\end{align}

where $\boldsymbol{Q}_{mn} \coloneqq \mathrm{diag}(\boldsymbol{q}_{mn})$. 
\end{definition}

\begin{figure}[h]
    \centering
    \begin{overpic}[width=0.3\textwidth,tics=5]{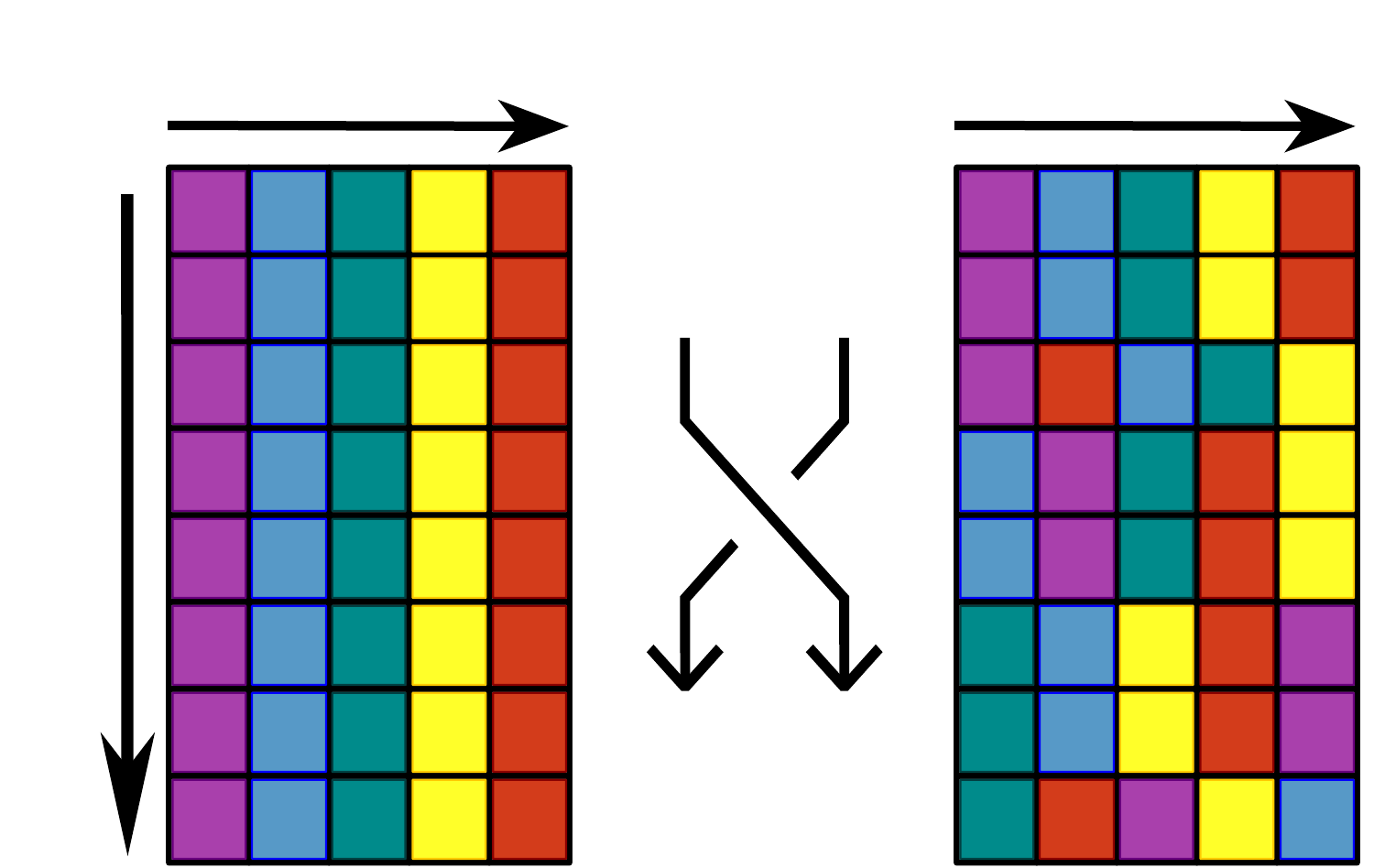}
    \put (24,57)  {$\bm x_m$}
    \put (81,57.3) {$\bm y_m$}
    \put (3,25) {$n$}
    \end{overpic}    
    \caption{If each channel of a multi-channel signal $\boldsymbol{x}_m$ is represented as a column in a data matrix, shuffling corresponds to independently permuting the entries within each row, resulting in the shuffled signals $\bm y_m$.}
    \label{fig:shuffled}
\end{figure}

{We give a slightly more general definition than that in \cite{Koka2024}:
{\begin{definition}[\textbf{$K$-Restricted Full Rank Property}]
\label{dfn:generic}
Let $K \le N$ and let $\Es \subset \Ce^N$ be a linear subspace of dimension $k\le K$. We say ``$\Es$ satisfies the \hyperref[dfn:generic]{$K$-restricted full rank property (RFRP)}" if $\dim \vartheta(\Es) =k$ for every coordinate projection $\vartheta: \Ce^N \rightarrow \Ce^N$ that preserves at least $K$ entries.
We say that ``a tall $N \times k$ matrix $\bm E$ satisfies the \hyperref[dfn:generic]{$K$-RFRP}", if its column space satisfies the \hyperref[dfn:generic]{$K$-RFRP}, i.e., if every $K \times k$ submatrix of $\bm E$ has rank $k$.\end{definition}}}

\begin{theorem}
\label{thm:CCULS}
{Consider $M$-channel signals $\bm x_m',\bm x_m''\in\Es_m,\:\Es_m\subset\mathbb{C}^N,\:m=1,\dots,M$, where $\bm \Es_m$ is a linear subspace for all $m$. Let $\{\bm Q_{m n}\}_{m,n=1}^M$ be binary diagonal matrices satisfying $\bm I=\sum_{m=1}^M\bm Q_{mn} = \sum_{n=1}^M\bm Q_{mn}$, and let $K=\underset{m,n}{\max}\{\dim(\Es_m+\Es_n)\}$. 
Suppose that i) the subspace $\Es_m + \Es_n$ satisfies the $K$-RFRP for every $m,n$, ii) $N \ge M K$ and iii) $\bm x_\kappa' \neq \bm x_\lambda'$ for every $\kappa \neq \lambda$.}
Then, the relation

\begin{align}
\begin{pmatrix}
        \boldsymbol{Q}_{11} & \dots&\boldsymbol{Q}_{1M} \\
        \vdots & \ddots&\vdots \\
        \boldsymbol{Q}_{M1} &\dots&\boldsymbol{Q}_{MM} \\
    \end{pmatrix}
    \begin{pmatrix}
    \bm x_1''\\
    \vdots\\
    \bm x_M''\\
\end{pmatrix} = 
\begin{pmatrix}
    \bm x_1'\\
    \vdots\\
    \bm x_M'\\
\end{pmatrix} \label{eq:M-channel-US}
\end{align}
implies $\bm x_m' = \bm x_{n_m}''$ for every $m$, with $\{n_m\}_{m=1}^M = \{m\}_{m=1}^M$. 
\end{theorem}
\begin{proof}
Let $\rho_{mn}: \mathbb{C}^N \rightarrow \mathbb{C}^N$ be defined as the coordinate projection via multiplication with $\bm{Q}_{mn}$, and let $r_{mn}$ denote the number of nonzero elements in $\bm{Q}_{mn}$. The $m$th row of the matrix equation \eqref{eq:M-channel-US} can be expressed as:

\begin{align}
\bm{x}_m' = \rho_{m1}(\bm{x}_1'') + \cdots + \rho_{mM}(\bm{x}_M''). \label{eq:M-channel-US-ith-row}
\end{align}

Noting that, for a fixed $m$, the projections $\{\rho_{mn}\}_{n=1}^M$ form an \emph{orthogonal resolution of the identity} \cite{Roman}, we may follow the same approach as in the proof of the first theorem in \cite{Koka2024} to arrive at relations of the form:  
\[
\rho_{m n_m}(\bm{x}_m' - \bm{x}_{n_m}'') = 0,
\]
\noindent {There must be at least one $n_m \in \{1, \dots, M\}$ such that $r_{m n_m} \ge K$, since $N \ge M K$.}
The above implies that $\bm{x}_m' - \bm{x}_{n_m}'' \in (\mathscr{E}_m + \mathscr{E}_{n_m}) \cap \ker(\rho_{m n_m})$. Let $\bm{E}_{mn_m}$ be a basis for the subspace $\mathscr{E}_m + \mathscr{E}_{n_m}$. {There exists a unique $\bm{\beta}$ such that $\bm{x}_m' - \bm{x}_{n_m}''= \bm{E}_{mn_m}\bm{\beta}$.} Let $\Ic \subset \mathscr{N}$ be the indices of the coordinates preserved by $\rho_{m n_m}$, and let $\bm{E}_{{mn_m},\Ic}$ be the row submatrix of $\bm{E}_{mn_m}$ that includes only the rows in $\bm{E}_{mn_m}$ indexed by $\Ic$. Similarly, for any vector $\bm{\zeta} \in \mathbb{C}^N$, let $\bm{\zeta}_\Ic$ denote the vector composed of elements of $\bm{\zeta}$ corresponding to coordinates in $\Ic$. Therefore, we have: {$\bm{0} = \rho_{m n_m}(\bm{x}_m' - \bm{x}_{n_m}'') = (\bm{x}_m' - \bm{x}_{n_m}'')_\Ic = \bm{E}_{{mn_m},\Ic} \bm{\beta}$. Now, $\bm{E}_{{mn_m},\Ic}$ has at least as many rows as columns, because $\dim(\mathscr{E}_m + \mathscr{E}_{n_m}) \le K$ while $|\Ic|\ge K$. It must then be that $\bm{E}_{{mn_m},\Ic}$ has full column-rank, since $\mathscr{E}_m + \mathscr{E}_{n_m}$ satisfies the \hyperref[dfn:generic]{$K$-RFRP}.} This leads to $\bm{\beta} = \bm{0}$ and thus $\bm{x}_m' = \bm{x}_{n_m}''$. As in \cite{Koka2024}, if it so happens that there exist indices $\kappa \neq \lambda$ such that $n_{\kappa} = n_{\lambda}$, we would have: $\bm{x}_\kappa' = \bm{x}_{n_\kappa}'' = \bm{x}_{n_\lambda}'' = \bm{x}_\lambda'.$  {However, this contradicts our hypothesis that $\bm{x}_\kappa' \neq \bm{x}_\lambda' $ for $\kappa \neq \lambda$.}
\end{proof}

Theorem~\ref{thm:CCULS} establishes conditions on the dimensions and number of shuffled samples that allow for unique recovery
in the cross-channel unlabeled sensing problem, when the underlying signals lie in a union of subspaces, {with the sum of any two subspaces} satisfying the
\hyperref[dfn:generic]{RFRP}. As we will see in the following two corollaries, Theorem~\ref{thm:CCULS} encompasses signal models falling under the compressed sensing framework as a special case, thus covering a wide variety of signal processing applications.
\begin{figure*}[b]
\centering
\includegraphics[width=0.85\linewidth]{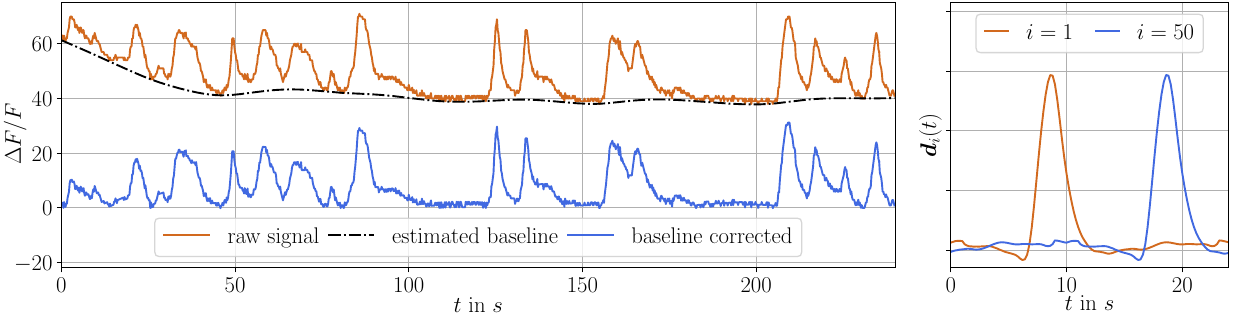}
    \caption{(Left) Baseline correction of the fluorescent traces $\Delta F/F$ (relative intensity) in Drosophila. (Right) Learned kernel using convolutional sparse coding.}
        \label{fig:atoms_dictionary}
\end{figure*}
\subsection{Cross-Channel Unlabeled Compressed Sensing}
As before, consider a multi-channel signal $\bm x_m$ for $m = 1, \dots, M$. This time, however, let $\bm x_m$ admit a $K_m$-sparse representation on an overcomplete dictionary $\bm D\in\mathbb{C}^{N\times p}$, i.e., $\bm x_m$ can be expressed as $\bm x_m = \bm D \bm \beta_m$, where $||\bm \beta_m||_0=K_m$. It is well-known that $\bm \beta_m$ is the unique solution to \[\underset{\bm \beta}{\min}||\bm \beta||_0 \text{ s. t. } \bm x_m = \bm D \bm \beta\] if $N\geq 2K_m$ and $\krank(\bm D)\geq2K_m$, where the $\krank(\bm D)$ is the largest number $\kappa$ such that any selection of $\kappa$ columns of $\bm D$ is linearly independent. To relate this signal model to Theorem~\ref{thm:CCULS}, we will rely on the following definition.
{ 
\begin{definition}[\bf{$K\times k$- and $K\times K$- Restricted Full Rank Properties}]
\label{dfn:krfrp}
Let $\mathbf{D}\in\mathbb{C}^{N\times p}$, where $p>N$ is allowed. We say "$\bm D$ satisfies the \hyperref[dfn:krfrp]{$K\times k$-restricted full rank property ($K\times k$-RFRP)}" if every $K\times k$ submatrix of $\bm D$ has rank $k$, $\forall k \le K$. We say $\bm D$ satisfies the $K\times K$-RFRP, if every $K \times K$ submatrix of $\bm D$ has rank $K$. Note $K\times k$-RFRP implies $K \times K$-RFRP.
\end{definition}}

\begin{remark}
    A random $N\times p$ matrix $\bm D$ will satisfy the \hyperref[dfn:krfrp]{$K\times k$-RFRP} for any $K\leq\min\{N,p\}$ with probability $1$, under any continuous probability distribution.
\end{remark}
 
\begin{corollary}
    \label{corr:1}
    Consider $M$ distinct signals $\{\bm x_m \in \Ce^N\}_{m=1}^M$ with $\bm x_m=\bm D\bm \beta_m,\:\bm\beta_m\in\mathbb{C}^p$, $\bm D\in\mathbb{C}^{N\times p}$. Denote by $\Ic_m$ the index set of nonzero coefficients in $\bm \beta_m$. Suppose that we observe the channels $\{\bm y_m\}_{m=1}^M$ of a shuffled M-channel signal with respect to the $\bm x_m$'s and that $||\sum_{m=1}^M\bm\beta_m||_0=|\bigcup_{m=1}^M\Ic_m|=K$. If $\bm D$ satisfies the {\hyperref[dfn:krfrp]{$2K\times 2K$-RFRP}} and the number of samples $N\geq MK$, then $\bm \beta_1,\dots,\bm\beta_M$ are uniquely determined up to a renaming of the channels.
    \end{corollary}

    \begin{proof}
    Since $\bm D$ satisfies the {\hyperref[dfn:krfrp]{$2K\times 2K$-RFRP}}, the vector $\bm\beta_\Sigma=\sum_{m=1}^M\bm\beta_m$, and consequently its index set of nonzero coefficients $\Ic=\bigcup_{m=1}^M\Ic_m$, are uniquely defined. This enables the construction of a sensing matrix $\bm E$ with column space $\Es$, derived by retaining the columns of $\bm D$ indexed by $\Ic$. Since $\bm x_m\in\Es$ for all $m$, and $\bm E$ satisfies the \hyperref[dfn:generic]{RFRP}, the result follows from Theorem~1 with $\Es_1 = \dots = \Es_m = \Es$.
    \end{proof}

\begin{corollary}
\label{corr:2}
 Consider $M$ distinct signals $\{\bm x_m \in \Ce^N\}_{m=1}^M$ with $\bm x_m=\bm D\bm \beta_m,\:\bm\beta_m\in\mathbb{C}^p$, $\bm D\in\mathbb{C}^{N\times p}$, $p > N$. Denote by $\Ic_m$ the index set of nonzero coefficients in $\bm \beta_m$. Suppose that we observe the channels $\{\bm y_m\}_{m=1}^M$ of a shuffled M-channel signal with respect to the $\bm x_m$'s and let $K = \underset{m,n}{\max}\{|\Ic_m\cup\Ic_n|\}$. If $\bm D$ satisfies the {\hyperref[dfn:krfrp]{$K\times k$-RFRP}} and the number of samples $N\geq MK$, then $\bm \beta_1,\dots,\bm\beta_M$ are uniquely determined up to a renaming of the channels.
\end{corollary}

\begin{proof}
{Each $\bm x_m$ lives in the subspace $\Es_m$ spanned by the columns of $\bm D$ indexed by $\Ic_m$; denote this column-submatrix of $\bm D$ by $\bm D^{\Ic_m}$. By construction, we have that $\dim(\Es_m+\Es_n) \le |\Ic_m\cup\Ic_n|:=K_{mn} \le K$. As $\bm D$ satisfies the \hyperref[dfn:krfrp]{$K \times k$-RFRP}, every $K \times K_{mn}$ submatrix of $\bm D$ has full column-rank. In particular, each $K \times K_{mn}$ submatrix of $\bm D^{\Ic_m \cup \Ic_n}$ has full column-rank. This means that i) $\bm D^{\Ic_m \cup \Ic_n}$ is a basis for the subspace $\Es_m+\Es_n$, and ii) $\Es_m+\Es_n$ satisfies the $K$-RFRP.} The result now follows directly from Theorem~\ref{thm:CCULS}.
\end{proof}

Corollaries~\ref{corr:1} and \ref{corr:2} establish that unique signal recovery from \( \bm y_m \) is achievable up to a permutation of the channel order under two conditions: (1) sufficient samples uniquely parameterize \( \bm x_m \) with respect to a specific set of columns from \(\bm D \), and (2) {the columns that the $\bm x_m$'s are supported on span subspaces that \emph{pairwise} satisfy the \hyperref[dfn:generic]{RFRP}}. These results suggest two algorithmic approaches: The first involves summing channel signals to identify a common subspace, followed by applying the cross-channel unlabeled sensing method from \cite{Koka2024} with a modified support recovery step for compressed sensing. This approach is examined in Section~\ref{sec:simul} using real data. However, summing signals can increase noise and subspace dimension, potentially requiring more measurements.
Alternatively, a joint approach could simultaneously address support recovery and signal estimation from \( \bm y_m \), potentially reducing sample size requirements, especially for large \( M \). This approach will be explored in future work.
\begin{figure*}
\centering
\includegraphics[width=0.98\linewidth]{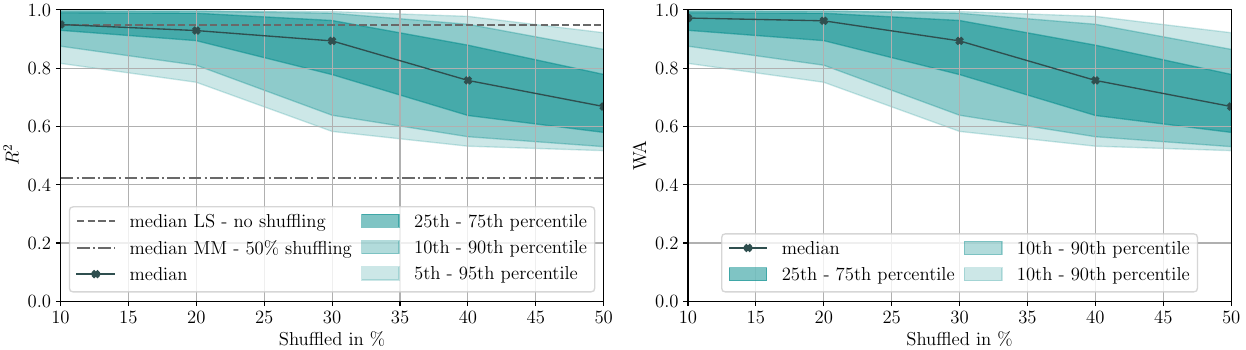}\\
    \caption{(Left) Performance of the method from \cite{Koka2024} in terms of $R^2$. Median performance of least squares (LS) estimator for no shuffling and MM-estimator for $50\%$ shuffling are shown as dashed and dashed-dotted lines. (Right) Performance of the method from \cite{Koka2024} in terms of $\mathrm{WA}$.}
    \label{fig:result}
\end{figure*}
\begin{figure}
\centering
\includegraphics[width=\linewidth]{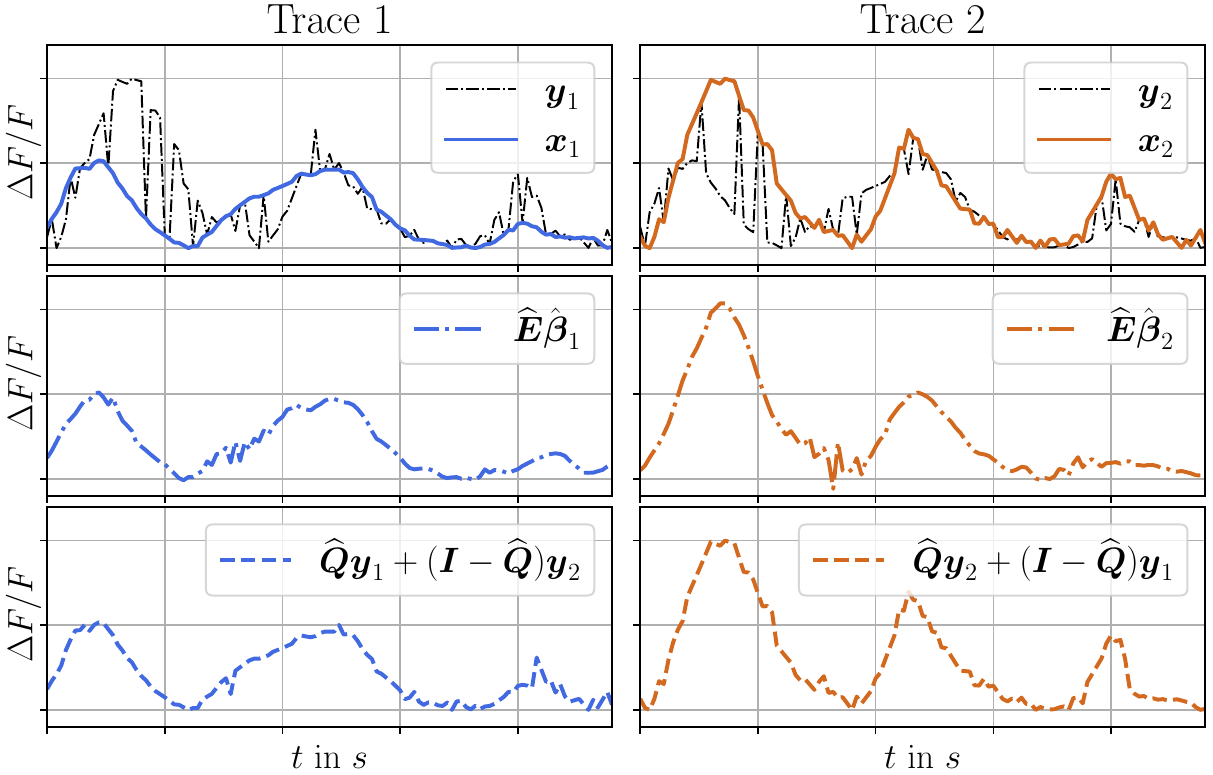}\\
    \caption{(Top) Original and shuffled traces $\bm x_1,\bm x_2$ and $\bm y_1,\bm y_2$, respectively. Amount of shuffled samples is $35\%$. (Middle) Reconstructed signals ($R^2 = 0.917$). (Bottom) Unshuffled measurements ($\mathrm{WA}=0.936$).}
    \label{fig:trace}
\end{figure}

%% file: Sections/simulations.tex
In our experiments, we evaluated the cross-channel unlabeled sensing framework using real calcium imaging traces from \emph{Drosophila} obtained from \cite{Streit2016}. Whole-brain calcium imaging captures neural activity by monitoring changes in calcium ion concentrations, which is particularly useful for studying neural circuits in freely moving organisms. However, movement can cause mismatches in sample-neuron assignments, complicating signal reconstruction and further analysis.

As a basic preprocessing, a baseline correction using asymmetric least squares smoothing as applied to the traces \cite{eilers2005baseline} (Fig.~\ref{fig:atoms_dictionary} left). Then, half of the available traces were used to learn a dictionary via convolutional sparse coding \cite{Grosse2007} and discarded for all further simulations. In particular, a single kernel $\bm{d}_1$ was learned, and the dictionary $\bm{D}$ was constructed as a circulant matrix of $\bm{d}_1$ (Fig.~\ref{fig:atoms_dictionary} right). The \hyperref[dfn:krfrp]{$K\times K$-RFRP} of $\bm D$ was numerically verified by checking the rank of $100\,000$ randomly sampled $K\times K$ submatrices for every $K\leq100$.

Sample-neuron mismatches were simulated by randomly shuffling pairs of $121$-sample ($\approx24$ s) long excerpts from different cells (denoted by $\bm x_1, \bm x_2$), resulting in the shuffled traces $\bm y_1, \bm y_2$. For the signal reconstruction and sample assignment task, the method from \cite{Koka2024}, with some modifications for our setup, was adopted:

\paragraph*{Sensing Matrix Estimation} The shuffled signals are summed as $\bm{y}_\Sigma = \bm{y}_1 + \bm{y}_2$, and used for the estimation of a sensing matrix $\bm{E}$. In particular, this is achieved by selecting columns from $\bm{D}$ using the least absolute shrinkage and selection operator (LASSO) \cite{tibshirani96}. The sparsity level was determined via stability selection \cite{Meinshausen2010,Shah2012}, which evaluates the stability of feature selection across subsamples and includes those with selection probabilities above a threshold. Here, we fix the selection threshold to $\alpha = 0.7$.

\paragraph*{Robust Estimation and Reassignment} Using the estimate $\widehat{\bm{E}}$ from the previous step, a robust MM-estimator of regression~\cite{yohai_high_1987} is applied to the matrix $\bm{I} \otimes \widehat{\bm{E}}$ with response vector $(\bm{y}_1^\top, \bm{y}_2^\top)^\top$ to estimate the coefficients $(\widehat{\bm{\beta}}_1^\top, \widehat{\bm{\beta}}_2^\top)^\top$. Samples are then reassigned between $\bm y_1$ and $\bm y_2$ to best fit the estimated models $\widehat{\bm{E}}\hat{\bm{\beta}}_1$ and $\widehat{\bm{E}}\hat{\bm{\beta}}_2$. This process is iterated for a predetermined number of times (in practice $5$ iterations are sufficient), with the best solution chosen based on the smallest residual sum of squares. An example of the results for a single run with $35\%$ shuffled samples is shown in Fig.~\ref{fig:trace}.

The performance was evaluated in terms of the $R^2$ value and a weighted accuracy (WA) metric tailored to measure sample assignment accuracy. The weighted accuracy is defined as:

\[
\mathrm{WA} = \frac{\sum_{\ell' \in \mathcal{L}} |x_{1\ell'} - x_{2\ell'}|}{\sum_{\ell=0}^{N-1} |x_{1\ell} - x_{2\ell}|},
\]

where $\mathcal{L}$ is the set of correctly assigned indices, weighted by the absolute deviations between $\bm{x}_1$ and $\bm{x}_2$. Given the ambiguity in channel order, the better of the two possible assignments was used. The performance metrics evaluated over $1000$ Monte Carlo runs are shown in Fig.~\ref{fig:result}, demonstrating the framework’s effectiveness in reconstructing shuffled calcium imaging signals. It can be observed that in terms of $R^2$, the estimates perform close to the least squares fit on the original trace, when up to $30\%$ of the samples are shuffled. Above $30\%$ shuffling there is a sharp decline in the performance, however even for $50\%$ shuffled samples we see a significant gain of reassigning the samples as compared to the robust MM-estimator. In terms of $\mathrm{WA}$, the performance shows a similar behavior.